\documentclass[conference,a4paper]{IEEEtran}

\newcommand{\extra}[1]{}

\newcommand{\comment}[1]{}

\usepackage{amsmath,amssymb,amsfonts}
\usepackage{amsthm}
\usepackage{epic,eepic,epsfig,latexsym,verbatim,revsymb}

\usepackage{mathdots}
\usepackage{enumerate}
\usepackage{braket}
\usepackage{calc}
\usepackage{float}
\usepackage{algorithmic}
\usepackage{slashed}
\restylefloat{figure}
\restylefloat{table}

\newtheorem{theorem}{Theorem}[section]

\newtheorem{conjecture}[theorem]{Conjecture}

\newtheorem{lemma}[theorem]{Lemma}

\theoremstyle{remark}
\newtheorem*{claim}{Claim}

\def\squareforqed{\hbox{\rlap{$\sqcap$}$\sqcup$}}
\def\qed{\ifmmode\squareforqed\else{\unskip\nobreak\hfil
\penalty50\hskip1em\null\nobreak\hfil\squareforqed
\parfillskip=0pt\finalhyphendemerits=0\endgraf}\fi}
\def\endenv{\ifmmode\;\else{\unskip\nobreak\hfil
\penalty50\hskip1em\null\nobreak\hfil\;
\parfillskip=0pt\finalhyphendemerits=0\endgraf}\fi}

\renewenvironment{proof}{\noindent \textbf{{Proof~} }}{\qed\medskip}
\newenvironment{proof+}[1]{\noindent \textbf{{Proof #1~} }}{\qed\medskip}

\mathchardef\ordinarycolon\mathcode`\:
\mathcode`\:=\string"8000
\def\vcentcolon{\mathrel{\mathop\ordinarycolon}}
\begingroup \catcode`\:=\active
  \lowercase{\endgroup
  \let :\vcentcolon
  }

\newcommand{\nc}{\newcommand}
\nc{\rnc}{\renewcommand}
\nc{\beq}{\begin{equation}}
\nc{\eeq}{{\end{equation}}}
\nc{\beqa}{\begin{eqnarray}}
\nc{\eeqa}{\end{eqnarray}}
\nc{\lbar}[1]{\overline{#1}}
\nc{\proj}[1]{| #1\rangle\!\langle #1 |}
\nc{\avg}[1]{\langle#1\rangle}
\nc{\smfrac}[2]{\mbox{$\frac{#1}{#2}$}}
\nc{\tr}{\operatorname{tr}}
\nc{\tracedist}[1]{\Delta_{}\!\left( #1 \right)}
\nc{\fid}[1]{F\!\left( #1 \right)}

\nc{\ox}{\otimes}
\nc{\dg}{\dagger}
\nc{\dn}{\downarrow}
\nc{\cA}{{\cal A}}
\nc{\cB}{{\cal B}}
\nc{\cC}{{\cal C}}
\nc{\cD}{{\cal D}}
\nc{\cE}{{\mathcal E}}
\nc{\cF}{{\cal F}}
\nc{\cG}{{\cal G}}
\nc{\cH}{{\cal H}}
\nc{\cI}{{\cal I}}
\nc{\cJ}{{\cal J}}
\nc{\cK}{{\cal K}}
\nc{\cL}{{\cal L}}
\nc{\cM}{{\cal M}}
\nc{\cN}{{\cal N}}
\nc{\cO}{{\cal O}}
\nc{\cP}{{\cal P}}
\nc{\cR}{{\cal R}}
\nc{\cS}{{\cal S}}
\nc{\cT}{{\cal T}}
\nc{\cU}{{\cal U}}
\nc{\cV}{{\cal V}}
\nc{\cX}{{\cal X}}
\nc{\cZ}{{\cal Z}}

\nc{\entI}{{\bf I}}
\nc{\entIarrow}{{\bf I}^{\leftarrow}}
\nc{\entH}{{\bf H}}
\nc{\entS}{{\bf S}}
\nc{\entHmin}{\mathsf{H}_{\min}}
\nc{\aentHmin}{\hat{\mathbf{H}}_{\min}}
\nc{\entHmins}{\mathbf{H}_{\min}^{\varepsilon}} % added by Lukas

\nc{\supp}{\textrm{supp}}
 
\nc{\entF}{{\bf E}_f}

\nc{\isom}{\simeq}

\nc{\rank}{\operatorname{rank}}
\nc{\rar}{\rightarrow}
\nc{\lrar}{\longrightarrow}
\nc{\polylog}{\operatorname{polylog}}
\nc{\poly}{\operatorname{poly}}
\nc{\1}{{\openone}}

\nc{\weight}{\textbf{w}}
\nc{\hamdist}{d_{H}}

\def\ve{\varepsilon}

\def\l{\lambda}

\nc{\Sp}{{{\mathbb S}}}
\nc{\RR}{{{\mathbb R}}}
\nc{\CC}{{{\mathbb C}}}
\nc{\FF}{{{\mathbb F}}}
\nc{\NN}{{{\mathbb N}}}
\nc{\ZZ}{{{\mathbb Z}}}
\nc{\PP}{{{\mathbb P}}}
\nc{\QQ}{{{\mathbb Q}}}
\nc{\UU}{{{\mathbb U}}}
\nc{\OO}{{{\mathbb O}}}
\nc{\EE}{{{\mathbb E}}}
\nc{\id}{{\operatorname{id}}}

\nc{\qubitchannel}{\id_2}
\nc{\bitchannel}{\overline{\id}_2}

\nc{\be}{\begin{equation}}
\nc{\ee}{{\end{equation}}}
\nc{\bea}{\begin{eqnarray}}
\nc{\eea}{\end{eqnarray}}
\nc{\<}{\langle}
\rnc{\>}{\rangle}
\nc{\Hom}[2]{\mbox{Hom}(\CC^{#1},\CC^{#2})}
\nc{\rU}{\mbox{U}}

\nc{\ob}[1]{#1}

\newcommand{\eqdef}	{:=}%{\stackrel{\textrm{def}}{=}}

\nc{\unif}{\textrm{unif}}

\nc{\circuit}{\textrm{circ}}
\nc{\haar}{\textrm{haar}}

\usepackage{color}

\begin{document}

\sloppy

%% Paper Title
%% You can use linebreaks \\ within to get better formatting as
%% desired. 
\title{On simultaneous min-entropy smoothing}

%% Author names and affiliations:
%%
%% Avoiding spaces at the end of the author lines is not a problem with
%% conference papers because we don't use \thanks or \IEEEmembership.
%%
%% For several authors with only one affiliation:
%%
% \author{
%   \IEEEauthorblockN{Hui-Ting Chang and Stefan M.~Moser}
%   \IEEEauthorblockA{Department of Electrical and Computer Engineering\\
%     National Chiao Tung University (NCTU)\\
%     Hsinchu, Taiwan\\
%     Email: \{email-of-hui-ting,email-of-stefan\}@ieee.org} 
% }
%%
%% For up to three affiliations:
%%
\author{
  \IEEEauthorblockN{Lukas Drescher}
  \IEEEauthorblockA{Institute for Theoretical Physics\\
    ETH Zuerich, Switzerland\\
    Email: lukasd@student.ethz.ch}
  \and
  \IEEEauthorblockN{Omar Fawzi}
  \IEEEauthorblockA{Institute for Theoretical Physics\\
    ETH Zuerich, Switzerland\\
    Email: ofawzi@phys.ethz.ch}
}
%%
%% For over three affiliations, or if they all won't fit within the width
%% of the page, use this alternative format:
%%
% \author{
%   \IEEEauthorblockN{
%     Michael Shell\IEEEauthorrefmark{1},
%     Homer Simpson\IEEEauthorrefmark{2},
%     James Kirk\IEEEauthorrefmark{3}, 
%     Montgomery Scott\IEEEauthorrefmark{3} and
%     Eldon Tyrell\IEEEauthorrefmark{4}}
%   \IEEEauthorblockA{
%     \IEEEauthorrefmark{1}School of Electrical and Computer Engineering\\
%     Georgia Institute of Technology, Atlanta, Georgia 30332--0250\\ 
%     Email: see http://www.michaelshell.org/contact.html}
%   \IEEEauthorblockA{
%     \IEEEauthorrefmark{2}Twentieth Century Fox, Springfield, USA\\
%     Email: homer@thesimpsons.com}
%   \IEEEauthorblockA{
%     \IEEEauthorrefmark{3}Starfleet Academy, San Francisco, California 96678-2391\\
%     Telephone: (800) 555--1212, Fax: (888) 555--1212}
%   \IEEEauthorblockA{
%     \IEEEauthorrefmark{4}Tyrell Inc., 123 Replicant Street, Los Angeles, California 90210--4321}
% }

%% Use for special paper notices

%% To balance the two columns, you should reduce the text-height of
%% the last page using the following command:
%%%%%%%%%%%%%%%%%%%%%%%%%%%%%%%%%%%%%%%%%%%%%%%%%%%%%%%%%%%%%%%%%%%%%
%\addtolength{\textheight}{-9.35cm}
%%%%%%%%%%%%%%%%%%%%%%%%%%%%%%%%%%%%%%%%%%%%%%%%%%%%%%%%%%%%%%%%%%%%%
%% with an appropriate value. This command must be place on the second
%% last page, i.e., for a one-page abstract here, for a two-page
%% abstract right after the \maketitle command.

%% Create the title:
\maketitle

%% Abstract: 
%% For the final version of the accepted paper, please make sure you
%% remove the comment "THIS PAPER IS ELIGIBLE FOR THE STUDENT PAPER
%% AWARD."
%%

\begin{abstract}
In the context of network information theory, one often needs a multiparty probability distribution to be typical in several ways simultaneously. When considering quantum states instead of classical ones, it is in general difficult to prove the existence of a state that is jointly typical. Such a difficulty was recently emphasized and conjectures on the existence of such states were formulated.
 
In this paper, we consider a one-shot multiparty typicality conjecture. The question can then be stated easily: is it possible to smooth the largest eigenvalues of all the marginals of a multipartite state $\rho$ simultaneously while staying close to $\rho$? We prove the answer is yes whenever the marginals of the state commute. 
%Moreover, we exhibit an example showing that the bound we obtain is tight. 
In the general quantum case, we prove that simultaneous smoothing is possible if the number of parties is two or more generally if the marginals to optimize satisfy some non-overlap property.
%that if the marginals of the state you want to optimize satisfy the property that two sets either are included in each other or are disjoint. We also discuss some obstacles that one encounters in proving such a result for general quantum states.
\end{abstract}

\section{Introduction}
It is natural in the context of studying information processing tasks to allow for a small error probability. This makes it possible to eliminate atypical behaviour of the system under consideration. When the state of a system is described by a probability distribution, an important quantity that arises in the analysis of information processing tasks is the largest probability. Events that happen with a probability that is atypically large can be discarded provided their total mass is smaller than the desired error probability. In order to optimize the rate of our task, one is then faced with an optimization over the choice of possible atypical sets. When the information processing task has multiple objectives, e.g., multiple receivers decoding the same message, there are several quantities to optimize. The question we consider here is how well can these different objectives be optimized simultaneously.

More concretely, consider a probability distribution $\rho$ on $m$ parties and fix some error tolerance $\ve > 0$. Each marginal has some largest probability. Given that an error probability $\ve$ is allowed, it is possible to discard atypical sets of weight at most $\ve$ in order to reduce the largest probability. This could be done separately for each marginal. Now is it possible to find a state on $m$ parties  that is still reasonably close to the original state $\rho$ but that is as good as the specific optimizers \emph{for all the marginals simultaneously}? As the optimization in this setting refers to eliminating atypical behaviour, we also refer to the process as ``smoothing''. %For quantum systems, the distribution is replaced by a positive semidefinite matrix whose eigenvalues correspond to probabilities. The operation of taking a marginal corresponds to a partial trace. In this quantum framework, a (classical) probability distribution is represented as a diagonal matrix with entries corresponding to the probabilities.
For quantum systems, the distribution is replaced by a positive semidefinite operator whose eigenvalues correspond to probabilities. The operation of taking a marginal corresponds to a partial trace. In this quantum framework, a (classical) probability distribution is represented as an operator with a particular eigenbasis.

One motivation for considering such a question is that it poses significant obstacles in the context of quantum network information theory as was recently emphasized in the study of multiparty state merging \cite{Dut11} and the study of the quantum interference channel \cite{FHSSW11, Sen11}.

The purpose of this paper is to formulate the questions that arose from these works in a one-shot setting. We provide a proof of the conjecture when certain commutation relations between the marginals of the state hold. We also give a proof for the two-party quantum case and when the marginals to optimize are ``non-overlapping''. These seem to be the cases that can be handled using the current techniques and we believe that new techniques are needed to prove the general case. We hope this work will raise interest in the conjecture and its cousins.

\section{Preliminary work}
\label{preliminary}

%- Notation: 
%\mathcal{H}: Hilbert space of pure states
%\mathcal{S}_=(\mathcal{H}): Space of normalized states
%\mathcal{S}_\leq(\mathcal{H}): Space of subnormalized states
%- Purified distance: definition, properties (bounds in terms of the trace distance, monotonicity under trace non-increasing maps (\ref{mt} Theorem 3.4), existence of equidistant extensions (\ref{mt} Corollary 3.6)
%- min-entropy and smooth version

\subsection{Basic notation}

The state of an isolated quantum system is represented by a unit vector in a Hilbert space. Quantum systems are denoted $A, A_1, A_2, \dots$ and are identified with their corresponding Hilbert spaces. We write $d_A \eqdef \dim A$. To describe a distribution $\{p_1, \dots, p_r\}$ over quantum states $\{\ket{\psi_1}, \dots, \ket{\psi_r}\}$ (also called a mixed state), we use a density operator $\rho = \sum_{i=1}^r p_i \proj{\psi_i}$. Here, $\proj{\psi}$ refers to the projector on the complex line spanned by $\ket{\psi}$. A density operator is a positive semidefinite operator with unit trace. 
%The density operator associated with a pure state is abbreviated by omitting the ket and bra $\psi \eqdef \proj{\psi}$. 
Let $\cP(A)$ be the set of positive semidefinite operators acting on $A$. Then $\cS(A) \eqdef \{ \rho \in \cP(A) : \tr \rho = 1\}$ is the set of density operators on $A$. The Hilbert space on which a density operator $\rho \in \cS(A)$ acts is sometimes denoted by a subscript, as in $\rho_A$. Superscripts are only used for labelling. In order to describe the state of a composite system $A_1A_2$, we use the tensor product Hilbert space $A_1 \otimes A_2$, which is sometimes simply denoted $A_1A_2$.  If $\rho_{A_1A_2}$ describes the joint state on $A_1A_2$, the reduced state on the system $A_1$ is obtained by the partial trace $\rho_{A_1} \eqdef \tr_{A_2} \rho_{A_1A_2}$. %In line with this, mixed states can also be interpreted as reduced states of a larger isolated quantum system which is in a pure state.

The evolution of any quantum system can be represented by a trace preserving completely positive map (TPCPM) $\cE_{A \to C}$. A map is called positive if for any positive operator $\rho$, $\cE(\rho)$ is also positive. It is called completely positive if for any quantum system $B$, the map $\cE \otimes \id_B$ is positive. 
%Every TPCPM can be written in an operator-sum representation, $\cE_{A \to C}(\rho) = \sum_{k} E_k \rho E_k^{\dagger}$, with Kraus operators $\{ E_k \}_k \subset \operatorname{Hom}(A,C)$ that satisfy $\sum_k E_k^{\dagger} E_k = \1_{A}$. 
For a map $\cE$ acting on system $A_1$, we sometimes drop an identity acting on another system, as in $\cE(\rho_{A_1A_2}) = (\cE \otimes \id_{A_2})(\rho_{A_1A_2})$.
%If $U$ is a unitary acting on $A$, and $\ket{\psi}$ a state in $A_1 \ox A_2$, we sometimes use $U \ket{\psi}$ to denote the state $(U \ox \1_B) \ket{\psi}$, where the symbol $\1_B$ is reserved for the identity map on $B$. 
For an introduction to quantum information, we refer the reader to \cite{NC00, Wil11}.

\subsection{Distance measures}
We use two distance measures based on extensions of the trace distance and the fidelity to subnormalized states,
$\cS_{\leq}(A) \eqdef \{\rho \in \cP(A) : \tr \rho \leq 1\}$.
%Subnormalized states themselves do not have an immediate physical interpretation. However, they can be seen as density operators of a larger system restricted to a subsystem. Moreover, every non-zero subnormalized density operator corresponds to an element of $\cS(A)$ through normalization. 
For subnormalized states, we define quantum evolutions as trace non-increasing completely positive maps. 
%In terms of Kraus operators we require only the inequality $\sum_k E_k^{\dagger} E_k \leq \1_A$ instead of equality. A more comprehensive introduction is given \cite[Chapter 3]{mt}. 

Let $\rho, \sigma \in \cS_{\leq}(A)$ be subnormalized density operators. The \emph{trace distance} is defined as
\[
D(\rho, \sigma) \eqdef \frac{1}{2} \| \rho - \sigma \|_1 + \frac{1}{2} | \tr( \rho - \sigma ) |,
\]
where $\| M \|_1 = \tr\left(\sqrt{M^{\dagger} M}\right)$. Another metric that is more commonly used in this context is the \emph{purified distance} \cite{TCR10},
\[
P(\rho, \sigma) \eqdef \sqrt{1-F(\rho,\sigma)^2},
\]
based on the \emph{generalized fidelity}, which is given by
\[
F(\rho, \sigma) \eqdef \| \sqrt{\rho} \sqrt{\sigma} \|_1 + \sqrt{(1 - \tr \rho)(1 - \tr \sigma)}.
\]
The two distance measures are related by
\begin{equation}
\label{eq:trace-purified}
D(\rho,\sigma) \leq P(\rho,\sigma) \leq \sqrt{2 D(\rho, \sigma)}.
\end{equation}
For the trace distance, the closed $\ve$-ball around $\rho$ is denoted by $B_{\ve}^D(\rho)$ and for the purified distance by $B_{\ve}^P(\rho)$. Quantum evolutions are non-expansive maps in both the trace distance and the purified distance. That is, for any trace non-increasing completely positive map $\cE$, we find
\begin{equation}
\begin{split}
D(\cE (\rho), \cE(\sigma)) &\leq D(\rho, \sigma), \\
P(\cE (\rho), \cE(\sigma)) &\leq P(\rho, \sigma).
\end{split}
\label{eq:monpd}
\end{equation}

\subsection{Min-entropy}
\label{minEntropy}
Let $\rho \in \cS_{\leq}(A)$. The min-entropy of the state $\rho$ is defined as
$\entHmin(A)_{\rho} \eqdef - \log \l_{\max}(\rho)$,
where $\l_{\max}(\rho)$ denotes the largest eigenvalue of $\rho$. Optimizing this quantity over an $\ve$-neighbourhood of $\rho$, we obtain the smooth min-entropy, 
\begin{equation}
\label{eq:smoothMinEntropy}
\begin{split}
\entHmin^{\varepsilon,X}(A)_{\rho} &\eqdef \max_{\sigma \in B_{\ve}^X(\rho)} \entHmin(A)_{\sigma} \\ &= - \log \left( \min_{\sigma \in B_{\ve}^X(\rho)} \l_{\max}(\sigma) \right)
\end{split}
\end{equation}
where $X$ can be set to either $D$ for trace-distance or $P$ for purified distance. As the purified distance is more common in this setting, we drop the superscript $P$ when using it. Since $B_{\varepsilon}^X(\rho) \subset \cS(A)$ is compact, the maximum in (\ref{eq:smoothMinEntropy}) is achieved by a state $\sigma \in B_{\varepsilon}^X(\rho)$,
\begin{equation}
\entHmin(A)_{\sigma} = \entHmin^{\varepsilon,X}(A)_{\rho}.
\label{eq:minEntropySmoother}
\end{equation}
% This state must necessarily satisfy $X(\rho, \sigma) = \ve$, since if it was in the interior of $B_{\ve}^X(\rho)$ we could multiply it by scalar $\l < 1$, with $\l \sigma \in B_{\ve}^X(\rho)$, to increase its min-entropy. Moreover, 
The state $\sigma$ can always be assumed to share a particular eigenbasis $\{ \ket{x} \}_x$ with $\rho$. \cite{mt} This follows from the fact that a measurement of $\sigma$ in this basis, $\cE: \sigma \mapsto \sum_x \proj{x} \sigma \proj{x}$, cannot increase the largest eigenvalue of $\sigma$. Since $\cE(\rho) = \rho$ we find by (\ref{eq:monpd}) that $\cE(\sigma) \in B_{\ve}^X(\rho)$. % It follows that when we have a single system $A$, the optimization problem \eqref{eq:smoothMinEntropy} is classical in the sense that it only depends on the eigenvalue distribution.
As a consequence when only considering a single system $A$, the optimization problem in \eqref{eq:smoothMinEntropy} is classical in the sense that we can always restrict it to states that share a particular eigenbasis with $\rho$.
%By unitary invariance of the distance measures $D$ and $P$, thus, the optimization problem (\ref{eq:smoothMinEntropy}) is equivalent to the classical case.
%We delay the discussion of this problem to chapter \ref{minEntropySmoothing} in favour of stating the main conjecture.

\section{Conjecture}
With this basic notation we state the main conjecture.
\begin{conjecture}
\label{conj:ssc}
For any number of parties $m \in \mathbb{N}$ there exists a function $g_m$ with $\lim\limits_{\ve \to 0} g_m(\ve) = 0$ such that the following holds.

For any state $\rho \in \mathcal{S}_{\leq}(A)$ on any $m$-party system $A = A_1 \cdots A_m$, there exists a state $\sigma \in B_{g_m(\ve)}^{P}(\rho)$ that satisfies
\[
\entHmin(S)_{\sigma} \geq \entHmin^{\varepsilon}(S)_{\rho}, \; \forall S \subset \{A_1,\dotsc,A_m\}, \; S \neq \emptyset.
\]
\end{conjecture}
The function $g_m$ can depend on the number of parties $m$ but it is important that it is independent of the physical realization. In particular, it must not depend on the dimensions of the systems $A_1, \dots, A_m$. Furthermore, note that by relation \eqref{eq:trace-purified} the conjecture can equivalently be restated for the trace distance.

This conjecture is a generalization of the multiparty typicality conjecture of \cite{Dut11} to general states that are possibly not tensor powers. As such, by an application of the asymptotic equipartition property \cite{TCR09}, Conjecture \ref{conj:ssc} directly implies the multiparty typicality conjecture. One could of course consider stronger versions of this conjecture and ask for the conditional entropies to be also simultaneously smoothed. And in fact, the simultaneous decoding conjecture in \cite{FHSSW11} would follow from an analogous conjecture with conditional entropies. As difficulties already arise without conditioning, we focus on this simple setting here.
% with the condition entropies but this is the simplest setting where the difficulty appears.

\section{Min-entropy smoothing}
\label{minEntropySmoothing}
In the following, we give an explicit formula for a state $\sigma \in B_{\ve}^D(\rho)$ that satisfies (\ref{eq:minEntropySmoother}) for $X = D$. This expression is then used to define smoothing as a quantum operation and to restate Conjecture \ref{conj:ssc} from a different perspective.

%As remarked in section \ref{minEntropy} the determination of the min-entropy smoothing state essentially is equivalent to the classical problem.
\begin{lemma}[Min-entropy smoothing, \cite{Dre13}]
\label{lem:minEntropySmoothing}
Let $\rho \in \cS_{\leq}(A)$, $\ve > 0$. Define the function
\begin{displaymath}
\tilde{f}_{\ve}(x) = \left\{ \begin{array}{ll}
2^{-\entHmin^{\ve,D}(A)_{\rho}}, \, &x > 2^{-\entHmin^{\ve,D}(A)_{\rho}} \\
x, \, &x \leq 2^{-\entHmin^{\ve,D}(A)_{\rho}}.
\end{array}  \right.
\end{displaymath}
Then the state $\sigma \eqdef \tilde{f}_{\ve}(\rho) \in B_{\ve}^{D}(\rho)$ satisfies \eqref{eq:minEntropySmoother} for $X = D$.
%The state $\sigma \in B_{\ve}^{P}(\rho)$ such that
%\begin{displaymath}
%\entHmin(A)_{\sigma} = \entHmin^{\ve,P}(A)_{\rho}
%\end{displaymath} 
%is unique and for $\l \eqdef \l_{\max}(\rho)$ with $m_{\l} \eqdef d_{\operatorname{Eig}(\rho, \l)}$, $\beta \eqdef \arccos( \sqrt{ \l m_{\l} } )$
%\begin{displaymath}
%\phi \eqdef \arcsin(\ve) \leq \arctan \left(\sqrt{\frac{\l_{\max -1}(\rho)}{\l_{\max}(\rho)}} \right) - \beta
%\end{displaymath}
%given by $\sigma = g_P(\rho)$ for 
%\begin{displaymath}
%f_{P,\ve}(x) = \left\{ \begin{array}{ll}
%\frac{\cos^2(\beta + \phi)}{ \cos^2(\beta) } x, \, &x > 2^{-\entHmin^{\ve,P}(A)_{\rho}} \\
%\frac{\sin^2(\beta + \phi)}{ \sin^2(\beta) } x, \, &x \leq 2^{-\entHmin^{\ve,P}(A)_{\rho}}
%\end{array}  \right.
%\end{displaymath}
\end{lemma}

Using this Lemma, we define $\entHmin$-smoothing as a quantum operation. Concisely, we realize it as a multiplication operator on the eigenvalues $\{ \l_i \}_i$ of the state $\rho \in \cS_{\leq}(A)$, mapping $\l_i$ to $f_{\ve}(\l_i) \l_i$. The smoothing function, $f_{\ve}(x) \eqdef \frac{\tilde{f}_{\ve}(x)}{x}$ for $x \in (0,1]$, $f_{\ve}(0) \eqdef 1$, is chosen according to Lemma \ref{lem:minEntropySmoothing}. Since $f_{\ve} \leq 1$, we can represent this map as a quantum operation on $\cS_{\leq}(A)$,
\begin{equation}
\cE: \tau \mapsto \sqrt{f_{\ve}}(\rho) \tau \sqrt{f_{\ve}}(\rho).
\label{eq:smoothingChannel}
\end{equation}
Note that this map is also a feasible smoothing operation for $P$ due to $\entHmin^{\ve,P}(A)_{\rho} \leq \entHmin^{\ve,D}(A)_{\rho}$ by (\ref{eq:trace-purified}). For the distance we then find $P(\rho, \cE(\rho)) \leq \sqrt{2 \varepsilon}$.

On a multiparty system $A = A_1 \cdots A_m$, this construction can be repeated on every subsystem $S$. For $\ve > 0$, we define a smoothing operation $\cE^S$ by
\begin{equation}
\cE^S: \begin{array}{ccc} \cS_{\leq}(S) & \to & \cS_{\leq}(S) \\ \tau_S & \mapsto & \sqrt{f_{\ve}^S}(\rho_S) \tau_S \sqrt{f_{\ve}^S}(\rho_S),
\end{array}
\label{eq:smoothingChannelSubsystem} 
\end{equation}
where the smoothing function $f_{\ve}^S$ is defined in terms of $\rho_S$. Conjecture \ref{conj:ssc} can then be restated as follows: can we construct a global quantum evolution $\cE: S_{\leq}(A) \to S_{\leq}(A)$ from the marginal smoothing operations $\{ \cE^{S} \}_{S \subset \{A_1, \dotsc, A_m \}, S \neq \emptyset}$ that simultaneously smooths all min-entropies of $\rho$ keeping $\rho$ close to itself? 

\section{Classical case}

%We show that classical states admit a natural solution of the problem of combining $\entHmin$-smoothing operations of different subsystems globally. We state this result in Theorem \ref{thm:classicalCase}. Furthermore, we determine the optimal function $g_m$ introduced in Conjecture \ref{conj:ssc}.

We show that classical states admit a natural solution of Conjecture \ref{conj:ssc} from the perspective of quantum evolutions. In particular, the smoothing operations $\{ \cE^S \}_{S \in \cK}$ \eqref{eq:smoothingChannelSubsystem} for the subsystems, once extended appropriately to the total system, can be combined to define an iterative simultaneous $\entHmin$-smoothing operation $\cE$. This result is stated in Theorem \ref{thm:classicalCase}. Furthermore, we provide a distribution showing the optimality of the obtained trace distance bound.
%$\{ g_m \}_{m \in \mathbb{N}}$.

Let $A = A_1 \cdots A_m$ be a classical system. A classical state $\rho \in \cS_{\leq}^{cl}(A)$ is characterized by its product eigenbasis,
\begin{equation}
\rho = \sum_{i_1 = 1}^{d_{A_1}} \cdots \sum_{i_m = 1}^{d_{A_m}} p_{i_1...i_m} \proj{i_1}_{A_1} \otimes \cdots \otimes \proj{i_m}_{A_m},
\label{eq:classicalstate}
\end{equation}
where $\{\ket{i_k}\}_{1 \leq i_k \leq d_{A_k}}$ denotes an orthonormal basis of $A_k$ for all $k \in \{1,\dotsc, m\}$.  Note that the structure of the classical state $\rho$ implies that its closest simultaneous $\entHmin$-smoother $\sigma$ can always be assumed to be classical. This follows from the fact that a measurement of $\sigma$ in the classical eigenbasis of $\rho$ can increase neither the largest eigenvalue of any of its reduced states nor the distance to $\rho$. Therefore, Conjecture \ref{conj:ssc} has a well-defined classical limit.

%In the following we will pursue an iterative approach to find a global simultaneous $\entHmin$-smoothing evolution $\cE$. We first extend the smoothing maps (\ref{eq:smoothingChannelSubsystem}) to act globally and showing in Lemma \ref{lem:classicalCase} that they satisfy several compatibility properties. 

%\lukas{The following might be considered trivial by some people... Instead: Extending the maps $\{ \cE^{S} \}_{S \in \cK}$ defined in \eqref{eq:smoothingChannelSubsystem} by $\bar{\cE}^S \eqdef \cE^S \otimes \id_{S^c}$ to act on $A$ the following Lemma shows that the smoothing operations $\{\bar{\cE}^S\}_{S \in \cK}$ are globally compatible.}

Extending the smoothing maps $\{ \cE^{S} \}_{S}$ defined in \eqref{eq:smoothingChannelSubsystem} to act globally by $\bar{\cE}^S \eqdef \cE^S \otimes \id_{S^c}$, so that $\cE^S(\rho_S) = \tr_{S^c}(\bar{\cE}^S (\rho))$, we observe the following properties. %the following Lemma shows that the smoothing operations $\{\bar{\cE}^S\}_{S \in \cK}$ are globally compatible.

%To extend the map $\cE^{S}$ defined in \eqref{eq:smoothingChannelSubsystem} to act on $A$, we note that instead of multiplying every eigenvalue of $\rho_S$ by the smoothing function $f_{\ve}^S$, by distributivity, we can equivalently multiply the eigenvalues of the global state $\rho$ that contribute to this marginal eigenvalue by the same value. In terms of operators this corresponds to $\cE^S(\rho_S) = \tr_{S^c}(\cE^S \otimes \id_{S^c} (\rho))$.
%The set of extensions obtained in this way satisfy the following properties.

%The set of extensions of marginal smoothing operations \eqref{eq:smoothingChannelSubsystem} obtained in this way satisfies the following compatibility properties.

\begin{lemma}
Let $\cK = 2^{\{ A_1, \dotsc, A_m \}} \setminus \{ \emptyset\}$. For a classical state $\rho \in S^{cl}_{\leq}(A)$ the extended smoothing operations $\{ \bar{\cE}^S \eqdef  \cE^S \otimes \id_{S^c} \}_{S \in \cK}$, $\cE^{S}$ as in (\ref{eq:smoothingChannelSubsystem}) for all $S \in \cK$, are
\begin{enumerate}[i)]
\item commutative: $\bar{\cE}^S \circ \bar{\cE}^T = \bar{\cE}^T \circ \bar{\cE}^S$ $\forall S, T \in \cK$,
\item density operator decreasing: $\bar{\cE}^S(\tau) \leq \tau$ $\forall \tau \in S_{\leq}^{cl}(A)$,
\item distance preserving: $D(\tau, \bar{\cE}^S(\tau) ) = D(\tau_S, \cE^S(\tau_S) )$ $\forall \tau \in S_{\leq}^{cl}(A)$.
\end{enumerate}
\label{lem:classicalCase}
\end{lemma}

\begin{proof}
i) The operators $\sqrt{f^S}(\rho_S) \otimes \1_{S^c}$, $\sqrt{f^T}(\rho_T) \otimes \1_{T^c}$ commute since $\rho_S \otimes \1_{S^c}$ and $\rho_T \otimes \1_{T^c}$ are simultaneously diagonalizable in the classical eigenbasis of $\rho$ for all $S, T \in \cK$. 

ii) This property holds for pure states $\tau = \proj{i_1}_{A_1} \otimes \cdots \otimes \proj{i_m}_{A_m}$ that span $\cS_{\leq}^{cl}(A)$. For $S = A_{r_1} \cdots A_{r_{|S|}}$ we have
%We find on $S \in \cK$
\[
%\bar{\cE}^S(\tau) = \underbrace{f^S(\tr( \tau_S \rho_S ))}_{\leq 1} \tau
\bar{\cE}^S(\tau) = \underbrace{f^S(\bra{i_{r_1},\dotsc,i_{r_{|S|}}} \rho_S \ket{i_{r_1},\dotsc,i_{r_{|S|}}})}_{\leq 1} \tau \leq \tau.
\]
By linearity of $\bar{\cE}^S$ this statement extends to all of $ \cS_{\leq}^{cl}(A)$.

iii) The trace distance simplifies to the trace for ordered density operators,
\begin{equation}
\tau, \omega \in \cS_{\leq}(A), \tau \geq \omega: D(\tau, \omega) = \tr(\tau - \omega),
\label{eq:gmtd}
\end{equation}
and, therefore, is independent of the subsystem where it is evaluated.
\end{proof}

By Lemma \ref{lem:classicalCase} the smoothing operations $\{ \cE^S \}_{S}$ from \eqref{eq:smoothingChannelSubsystem} can be globally combined in a compatible way giving rise to
\begin{theorem}[Classical case of Conjecture \ref{conj:ssc}]
Let $\rho \in \cS_{\leq}^{cl}(A)$, $\cK \subset 2^{\{A_1,...A_m\}} \setminus \{ \emptyset \}$, $\varepsilon > 0$. There exists a state $\sigma \in \cS_{\leq}^{cl}(A)$ that satisfies
\begin{eqnarray}
&\entHmin(S)_{\sigma} &\geq \entHmin^{\varepsilon,D}(S)_{\rho} \text{ } \forall S \in \cK  \label{eq:ccss} \\
&D(\rho,\sigma) &\leq |\cK| \varepsilon. \label{eq:comdist}
\end{eqnarray}
In general, the bound \eqref{eq:comdist} is optimal in the limit of large dimensions $\min\limits_{1 \leq i \leq m} d_{A_i}$. 
\label{thm:classicalCase}
\end{theorem}
%\begin{theorem}[Classical case of Conjecture \ref{conj:ssc}]
%Let $\rho \in \cS_{\leq}^{cl}(A)$, $\cK \subset 2^{\{A_1,...A_m\}} \setminus \emptyset$, $\varepsilon > 0$. Let $\{ \bar{\cE}^S \eqdef \cE^S \otimes \id_{S^c} \}_{S \in \cK}$ be extensions of the smoothing operations defined in (\ref{eq:smoothingChannelSubsystem}), let $(S^i)_{1 \leq i \leq |\cK|}$ be an arbitrary ordering of $\cK$. Then the iteratively smoothed state $\sigma \eqdef \bar{\cE}^{S^1} \circ \cdots \circ \bar{\cE}^{S^{|\cK|}}(\rho) \in \cS_{\leq}^{cl}(A)$ satisfies
%\begin{eqnarray}
%&\entHmin(S)_{\sigma} &\geq \entHmin^{\varepsilon,D}(S)_{\rho} \text{ } \forall S \in \cK  \label{eq:ccss} \\
%&D(\rho,\sigma) &\leq |\cK| \varepsilon. \label{eq:comdist}
%\end{eqnarray}
%\label{thm:classicalCase}
%\end{theorem}

\begin{proof}
Let $(S^i)_{1 \leq i \leq |\cK|}$ be an arbitrary ordering of the set $\cK$. Define the iteratively smoothed state
\[
\sigma \eqdef \bar{\cE}^{S^1} \circ \cdots \circ \bar{\cE}^{S^{|\cK|}}(\rho) \in \cS_{\leq}^{cl}(A).
\]
Since property ii) Lemma \ref{lem:classicalCase} carries over to any concatenation of the maps $\{ \bar{\cE}^{S} \}_{S \in \cK}$ it follows that
\begin{displaymath}
\sigma \leq \bar{\cE}^{S^1} \circ \cdots \circ \bar{\cE}^{S^{i}}(\rho) \leq \bar{\cE}^{S^{i}}(\rho)
\end{displaymath}
using complete positivity of $\cE^{S}$, $\forall S \in \cK$, in the first step. This relation inherits to the subsystem $S^i$ under the partial trace, where it becomes $\sigma_{S^i} \leq \cE^{S^i}(\rho_{S^i})$, thus implying (\ref{eq:ccss}).

To bound the distance we successively apply the triangle inequality,
\begin{displaymath}
\begin{split}
D(\rho,\sigma) &\leq \sum\limits_{i=1}^{|\cK|} D(\bar{\cE}^{S^{1}} \circ \cdots \circ \bar{\cE}^{S^{i-1}} (\rho), \bar{\cE}^{S^{1}} \circ \cdots \circ \bar{\cE}^{S^{i}} (\rho)) \\ &\overset{\eqref{eq:monpd}}{\leq} \sum\limits_{i=1}^{|\cK|} D(\rho, \bar{\cE}^{S^{i}} (\rho)) \overset{\text{iii)}}{=} \sum\limits_{i=1}^{|\cK|} D(\rho_{S^i}, \cE^{S^{i}} (\rho_{S^i})) \leq |\cK| \ve
\end{split}
\end{displaymath}
where we have used Lemma \ref{lem:classicalCase}, iii), in the third step.

% optimality example

We prove that the bound (\ref{eq:comdist}) is optimal for two parties. The general case can be found in \cite{Dre13}. Let the parties $A_1$, $A_2$ have equal dimension, $d_{A_1} = d_{A_2} = 2 n^2 + 1$ for $n \in \mathbb{N}$.  Define a state $p$ on the register $\{1,\dotsc,d_{A_1}\} \times \{1,\dotsc,d_{A_2}\}$ by the probability distribution
\begin{equation} p = \left( \begin{array}{ccccccc}
                    &        &            & \frac{f_{A_2}}{2n^2} &  	           &                &          \iddots \\
                    &        &            &        \vdots       &              &  \frac{f_{A_1A_2}}{2n} &  \\
                    &        &            & \frac{f_{A_2}}{2n^2} &  \iddots     &       & \\
\frac{f_{A_1}}{2n^2} & \ldots & \frac{f_{A_1}}{2n^2}  & 0               & \frac{f_1}{2n^2} & \ldots & \frac{f_1}{2n^2} \\
                    &        &   \iddots  & \frac{f_{A_2}}{2n^2} &              &                 &                 \\
          & \frac{f_{A_1A_2}}{2n} &    &   \vdots       &  	           &                 &                 \\
      \iddots       &        &            & \frac{f_{A_2}}{2n^2} &              &                 &                 \\
   \end{array} \right),
\label{eq:ex}
\end{equation}
where only every $n$-th entry on the diagonal is occupied by $\frac{f_{A_1A_2}}{2n}$. All blank entries are set to $0$. Let $\cK \subset \{ A_1, A_2, A_1 A_2 \}$. For $S \in \cK$ define $f_S \eqdef \frac{1}{|\cK|}$, else set $f_S \eqdef 0$.

\begin{claim}
For every $\ve < \frac{1}{|\cK|}$ there exists $n_0 \in \mathbb{N}$ so that $\forall n \geq n_0$ any classical state $q$ on $A_1A_2$ with
\begin{equation}
\entHmin(S)_{q} \geq \entHmin^{\varepsilon,D}(S)_{p} \text{ } \forall S \in \cK
\label{eq:clex}
\end{equation}
satisfies $D(p,q) \geq |\cK| \ve$.
\end{claim}

To prove this claim we denote the horizontal non-zero line in (\ref{eq:ex}) by $h^{A_1}$, the vertical non-zero line by $h^{A_2}$ and the non-zero diagonal by $h^{A_1 A_2}$.
Computing the marginals to
\begin{displaymath}
(p_{A_j})_i \left\{ \begin{array}{ll} = f_{A_j} & \text{if } i = n^{2} \\ \leq \frac{f_{A_1A_2}}{2n}+\frac{f_{A_{\{1,2\} \setminus \{ j \}}}}{2n^2} & \text{else.}\end{array} \right.
\end{displaymath}
we observe that the entries of $p_S$, $S \in \cK$, coming from $h^{S}$ dominate all others by order $n$. Hence, for any $\ve < \frac{1}{|\cK|}$ there exists an $n_0$ so that $\forall n \geq n_0$ a probability weight of at least $\ve$ has to be removed from $h^S$ in order to smooth $p$ on $S$. Since the only common entry of the sets $\{ h^S \}_{S \in \cK}$ has probability $0$ the claim follows.
\end{proof}

The construction of $p$ in (\ref{eq:ex}) can be naturally generalized to $m$ parties. The probability distribution $p$ is then defined on the discrete $m$-cube. The discrete lines $h^S$ are replaced by discrete hyperplanes, each lying orthogonal to the main diagonal of the subspace associated to the subsystem $S$. The density of non-zero entries on these hyperplanes decreases exponentially in the number of parties in subsystem $S$. The calculations are somewhat more involved and can be found in \cite{Dre13}. 

%Applying the previous lemma to the problem of simultaneous smoothing of the min-entropy gives
%\begin{theorem}
%Let $\rho \in \cS_{\leq}(A)$, $\cK \subset 2^{\{A_1,...A_m\}} \setminus \emptyset$ and  $\varepsilon > 0$. For $S \in \cK$ define $\Pi^S \in \cP(S)$, $\Pi^S \leq \1_S$ such that
%\begin{displaymath}
%\begin{split} 
%&\entHmin(S)_{\Pi^S \rho \Pi^S} = \entHmin^{\varepsilon,D}(S)_{\rho}, \\
%&D(\Pi^S \rho_S  \Pi^S, \rho_S) \leq \varepsilon.
%\end{split}
%\end{displaymath}
%If $\{\Pi^S\}_{S \in \cK}$ and $\rho$ fulfill conditions (\ref{eq:crpp}) and (\ref{eq:crpr}) then $\sigma \in \cS_{\leq}(A)$ defined as in (\ref{eq:ccsigma}) satisfies
%\begin{eqnarray}
%&\entHmin(S)_{\sigma} &\geq \entHmin^{\varepsilon,D}(S)_{\rho} \text{ } \forall S \in \cK  \label{eq:ccss} \\
%&D(\rho,\sigma) &\leq |\cK| \varepsilon. \label{eq:comdist}
%\end{eqnarray}
%\label{thm:classicalCase}
%\end{theorem}

%\begin{proof}
%Apply Lemma \ref{lem:com} and use the monotonicity of the smooth min-entropy (\ref{lem:monh}) for $\sigma \leq \rho$ to obtain (\ref{eq:ccss}).
%\end{proof}

Choosing $\cK = 2^{\{A_1,...A_m\}} \setminus \{ \emptyset \}$ Theorem \ref{thm:classicalCase} proves conjecture \ref{conj:ssc} for classical states with a trace distance bound of $(2^m-1)\ve$ that is optimal (for trace distance smoothing) as shown by the distribution described in the previous paragraph.
In fact, a modified version of Lemma \ref{lem:classicalCase}, where ii) and iii) only hold for $\tau = \rho$, applies to any state $\rho \in \cS_{\leq}(A)$ that satisfies the commutation relations
\begin{equation}
[\rho_{S} \otimes \1_{S^c},\rho_{T} \otimes \1_{T^c}] = 0 \text{ } \forall S, T \in \cK.
\label{eq:commmutationRelations}
\end{equation}
This is sufficient to prove Theorem \ref{thm:classicalCase} \cite{Dre13}. As a non-classical example that satisfies (\ref{eq:commmutationRelations}) consider a bipartite entangled pure state $\Ket{\psi} =  \sum_{j=1}^{d} \frac{1}{\sqrt{d}} \Ket{j}_{A_1} \otimes \Ket{j}_{A_2}$. 

Finally, considering simultaneous smoothing in the purified distance, we remark that the state $p$ \eqref{eq:ex} and its generalizations to $m$ parties for non-singleton $\cK \subset 2^{\{A_1,...A_m\}} \setminus \{ \emptyset \}$ have a closest simultaneous $\entHmin$-smoother $q$ that satisfies
\[
P(p,q) = \sqrt[4]{|\cK|-1} \sqrt{2 \ve} + \cO(\ve)
\]
in the limit $(\ve \to 0)$ \cite{Dre13}. This shows that a square-root dependence in $\ve$ is unavoidable when simultaneously smoothing in the purified distance.

\section{Quantum case}

We start by analyzing the differences of the quantum case to the classical setting. Focussing on property ii) in Lemma \ref{lem:classicalCase} we may ask: for any $\rho \in \cS_{\leq}(A)$, does there exist a close state $\sigma \leq \rho$ with $\entHmin(S)_{\sigma} \geq  \entHmin^{\ve,D}(S)_{\rho}$? The existence of such a state would immediately yield a proof for the quantum case of Conjecture \ref{conj:ssc} by the fact that the smooth min-entropy is monotonous in the positive semidefinite ordering on $\cS_{\leq}(A)$. It turns out, however, that in general the answer is negative. As a counterexample consider a pure state $\rho$, so that one of its marginals $\rho_S$ is almost fully mixed with the exception of one eigenvalue, which is $\ve$ larger than all others. The state $\sigma$ by $\sigma \leq \rho$ must then be a multiple of $\rho$, the best possible proportionality factor being $2^{-(\entHmin^{\ve}(S)_{\rho} - \entHmin(S)_{\rho})}$, which tends to $0$ as $(d_S \to \infty)$.

Returning to the quantum evolution perspective, the extended smoothing operations $\{ \bar{\cE}^S \eqdef \cE^S \otimes \id_{S^c} \}_{S \in \cK}$, where $\cE^S$ is defined as in \eqref{eq:smoothingChannelSubsystem}, will in general not satisfy Lemma \ref{lem:classicalCase}, iii). Instead they satisfy the same property in the purified distance. 
\begin{lemma}
\label{lem:purified}
Let $\tau \in \mathcal{S}_{\leq}(A_1A_2)$, $\Pi^{A_1} \in \mathcal{P}(A_1)$, $\Pi^{A_1} \leq \1_{A_1}$, such that $[\Pi^{A_1},\tau_{A_1}] = 0$. Then
\begin{equation}
P(\tau, \Pi^{A_1} \tau \Pi^{A_1}) = P(\tau_{A_1}, \Pi^{A_1} \tau_{A_1} \Pi^{A_1}).
\label{eq:purified}
\end{equation}
\end{lemma}
\begin{proof}
The inequality ``$\geq$" follows by the monotonicity property of the purified distance (\ref{eq:monpd}) under the TPCPM $\tr_{S^c}$. \cite{mt} To derive the other inequality we use Uhlmann's Theorem for the fidelity \cite{uhlmann}. Let $\Ket{\psi} \in A_1 A_2 R$ be a purification of $\tau$, then
\begin{displaymath}
\begin{split}
\| \sqrt{\tau} \sqrt{\Pi^{A_1} \tau \Pi^{A_1}} \|_1 &\geq \Bra{\psi} \Pi^{A_1} \Ket{\psi} \\ &= \tr(\Pi^{A_1} \tau_{A_1}) \\ & \overset{[\Pi^{A_1},\tau_{A_1}] = 0}{=} \|\sqrt{\tau_{A_1}} \sqrt{\Pi^{A_1} \tau_{A_1} \Pi^{A_1}} \|_1
\end{split}
\end{displaymath}
where in the first line it was used that $\Pi^{A_1} \Ket{\psi}$ is a purification of $\Pi^{A_1} \tau \Pi^{A_1}$. As $\tr(\Pi^{A_1} \tau \Pi^{A_1}) = \tr(\Pi^{A_1} \tau_{A_1} \Pi^{A_1})$ and $\tr (\tau) = \tr (\tau_{A_1})$ we conclude
\begin{displaymath}
F(\tau, \Pi^{A_1} \tau \Pi^{A_1}) = F(\tau_{A_1}, \Pi^{A_1} \tau_{A_1} \Pi^{A_1}).
\end{displaymath}
\end{proof} 

%A similar statement can be proved for the trace-distance. In fact, it has already been used in the proof of Lemma \ref{lem:com}. However, it requires the stronger commutation relation $[\Pi^B, \tau] = 0$ to hold while the purified distance only asks for commutativity of the considered operators on the subsystem.

Using this Lemma, we show that the construction from the previous chapter can be transferred to the quantum setting yielding a proof of Conjecture \ref{conj:ssc} for two parties.
%This is done by a iteratively smoothing the min-entropy of all subsystems according to an appropriate order. The techniques admit a natural generalization to the setting with more parties solving a restricted form of conjecture \ref{conj:ssc} in this case. As a consequence one can conclude conjecture \ref{conj:ssc} for a three party system in a pure state. Finally, we analyze the simplest case where the iterative approach fails to provide a simultaneous min-entropy smoother, which is the three party system in a mixed state.

\subsection{Two parties}

%In this section we give proof for conjecture \ref{conj:ssc} in the case of two parties consisting of quantum systems $A_1$ and $A_2$. 
We note that the multiparty typicality conjecture, which is the special case when $\rho$ is a tensor power state, was proved in \cite{Dut11} and subsequently in \cite{Noe12} for two parties. We provide here a proof in the more general one-shot setting which is hopefully more transparent.
\begin{theorem}[Quantum case of conjecture \ref{conj:ssc} for two parties]
\label{thm:quant2p}
Let $\rho \in \cS_{\leq}(A_1A_2)$, $\cK \subset \{ A_1, A_2, A_1 A_2 \}$, $\varepsilon > 0$. There exists $\sigma \in \cS_{\leq}(A_1A_2)$ such that
\begin{eqnarray}
\entHmin(S)_{\sigma} &\geq &\entHmin^{\varepsilon}(S)_{\rho} \text{ } \forall S \in \cK, \label{eq:2pss}\\
P(\rho,\sigma) &\leq & |\cK| \sqrt{2 \varepsilon}. \label{eq:2pdist}
\end{eqnarray}
\end{theorem}

The proof requires the following basic lemma.
% which is a generalization of the principle of locality for normalized states.

\begin{lemma}
Let $\rho \in \cS_{\leq}(A_1 A_2)$, $\cE_{A_2 \to A_2}$ be a quantum evolution on $A_2$. Then,
$(\id_{A_1} \otimes \cE_{A_2 \to A_2}(\rho))_{A_1} \leq \rho_{A_1}$.
\label{lem:principleOfLocality}
\end{lemma}

We omit the proof of this basic fact, being essentially a consequence of the cyclicity of the partial trace in operators acting only on the system traced out.

%\begin{IEEEproof}[Proof of Lemma \ref{lem:principleOfLocality}]
%Let $\cE^{A_2}(\cdot) = \sum_k E_k \cdot E_k^{\dagger}$ be an operator-sum representation of $\cE^{A_2}$, where $\{E_k\}$ is a collection of endomorphisms on $A_2$ satisfying $\sum_k E_k^{\dagger} E_k \leq \1_{A_2}$. By the cyclicity of the partial trace in the system traced out, we then find
%\begin{displaymath}
%\begin{split}
%(\id_{A_1} \otimes \cE^{A_2}(\rho))_{A_1} &= \tr_{A_2}(\1_{A_1} \otimes \sum_k  E_k^{\dagger} E_k \rho) \\
%\end{split}
%\end{displaymath}
%and thus setting $\Pi_{A_2} = \1_{A_2} - \sum_k  E_k^{\dagger} E_k \leq \1_{A_2}$
%\begin{displaymath}
%\begin{split}
%\rho_{A_1} - (\id_{A_1} &\otimes \cE^{A_2}(\rho))_{A_1} \\ &= \tr_{A_2}(\1_{A_1} \otimes \Pi_{A_2} \rho) \\ &= \tr_{A_2}(\1_{A_1} \otimes \sqrt{\Pi_{A_2}} \rho \1_{A_1} \otimes \sqrt{\Pi_{A_2}}) \\ &\geq 0.
%\end{split}
%\end{displaymath}
%\end{IEEEproof}

\begin{IEEEproof}[Proof of Theorem \ref{thm:quant2p}]
We define $\cE^{S}$ as in (\ref{eq:smoothingChannelSubsystem}) for $S \in \cK$, $\cE^{S} \eqdef \id_S$ else, and $\bar{\cE}^{S} \eqdef \cE^{S} \otimes \id_{S^c}$. Choose the order $(S^1,S^2,S^3) = (A_1,A_2,A_1A_2)$. Define
\begin{displaymath}
\sigma \eqdef \cE^{A_1} \circ \cE^{A_2} \circ \cE^{A_1 A_2}(\rho).
\end{displaymath} 
This state has the right min-entropies (\ref{eq:2pss}):
\begin{itemize}
\item On the total system $A_1 A_2$ we can apply the submultiplicativity of $\| \cdot \|_{\infty}$:
\begin{displaymath}
\| \sigma \|_{\infty} \leq \underbrace{\| \sqrt{f_{\ve}^{A_1}} \|_{\infty}^2}_{\leq 1} \underbrace{\| \sqrt{f_{\ve}^{A_2}} \|_{\infty}^2}_{\leq 1} \| \bar{\cE}^{A_1 A_2}(\rho) \|_{\infty}
\end{displaymath}
since $f_{\ve}^{A_i} \leq 1$, $i=1,2$.
\item On subsystem $A_1$ we have
\begin{displaymath}
\begin{split}
\sigma_{A_1} &= \cE^{A_1} \circ \tr_{A_2}(\bar{\cE}^{A_2} \circ \underbrace{ \bar{\cE}^{A_1A_2}(\rho) }_{\leq \rho}) \\ &\leq \cE^{A_1} \circ \underbrace{\tr_{A_2}(\bar{\cE}^{A_2} \rho)}_{\leq \rho_{A_1}} \leq \cE^{A_1}(\rho_{A_1})
\end{split}
\end{displaymath}
using Lemma \ref{lem:principleOfLocality} in the last step. Since $\bar{\cE}^{A_1}$ and $\bar{\cE}^{A_2}$ commute the same argument applies on $A_2$.
\end{itemize} 
The distance part is entirely analogous to the classical case (cf. proof of Theorem \ref{thm:classicalCase}). The only difference is that $P$ is used here throughout instead of $D$. Accordingly, Lemma \ref{lem:purified} substitutes Lemma \ref{lem:classicalCase}, iii). Recalling that $\cE^{S^i}$ was designed to smooth in the trace distance, in the last step we use \eqref{eq:trace-purified} to obtain a bound on the purified distance, $P(\rho_{S^i},\cE^{S^i}(\rho_{S^i})) \leq \sqrt{2\ve}$.
\end{IEEEproof}

\subsection{Non-overlapping subsystems}

The proof of Theorem \ref{thm:quant2p} can be generalized to an $m$-party system $A = A_1 \cdots A_m$ where the subsystems $\cK$ under consideration can be ordered with respect to the inclusion. 
%That is, there exists an ordering $(S^i)_{1 \leq i \leq |\cK|}$ that satisfies
%\begin{displaymath}
%1 \leq i < j \leq |\cK| \Rightarrow S^i \subset S^j \vee S^i \cap S^j = \emptyset. 
%\end{displaymath}
The subsystems are then iteratively smoothed according to such an order starting with the largest system.
%More precisely, if $\cK$ is such that for all $S, T \in \cK$, we have either $S \cap T = \emptyset$, $S \subset T$ or $T \subset S$, then it is possible to simultaneously smooth all the marginals defined by $\cK$. 
Due to space limitations, we omit the proof of this result here.

\begin{theorem}[Quantum case of conjecture \ref{conj:ssc} for non-overlapping subsystems]
\label{thm:non-overlapping}
Let $\rho \in \cS_{\leq}(A)$, $\varepsilon > 0$. Let $\cK \subset 2^{\{A_1,\dotsc,A_m\}} \setminus \{ \emptyset \}$ be such that
\begin{equation}
\forall S, T \in \cK: (S \subset T)  \vee (T \subset S) \vee (S \cap T = \emptyset).
\label{eq:non-overlapping}
\end{equation}
%Let $(S^i)_{i \in \{1, \dotsc, |\cK| \}}$ be an ordering of $\cK$ that respects the inclusion, $S^i \slashed \supset S^j$ if  $i \leq j$ for $i,j \in \{1, \dotsc, |\cK| \}$. Define $\Pi^S \in \cP(S)$, $\Pi^S \leq \1_S$ for $S \in \cK$ such that in decreasing order in $j = |\cK|,\dotsc,1$
%\begin{eqnarray}
%\entHmin(S^j)_{\Pi^{S^j} \sigma^{S^j} \Pi^{S^j}} &= \entHmin^{\varepsilon,D}(S^j)_{\sigma^{S^j}}, \label{eq:mpsm} \\
%D(\sigma^{S^j}, \Pi^{S^j} \sigma^{S^j} \Pi^{S^j}) &\leq \varepsilon,
%\label{eq:dmp}
%\end{eqnarray}
%holds, where $\sigma^{S^j} \eqdef \prod_{i = j+1}^{|\cK|} \rho (\prod_{i = j+1}^{|\cK|})^{\dagger}$ $\forall j \in \{1,\dotsc, |\cK| - 1\}$, $\sigma^{S^{|\cK|}} = \rho$. 
There exists a state $\sigma$ that satisfies
$
\entHmin(S)_{\sigma} \geq \entHmin^{\varepsilon}(S)_{\rho} \text{ } \forall S \in \cK$
and 
$
P(\rho,\sigma) \leq |\cK| \sqrt{2 \varepsilon}.$
\end{theorem}

%We shall give a sketch of the proof.
%\begin{proof}
%Let $S^j \in \cK$ be a fixed subsystem. To show property (\ref{eq:mpss}) of the state $\sigma$ for $S^j$, note that all $S^i$, $1 \leq i \leq j-1$ are either subsystems of $S^j$ or do not intersect $S^j$. In the first case, we apply the submultiplicativity property of $\| \cdot \|_{\infty}$ as in (\ref{eq:2psmabproof}) to show that the corresponding $\Pi^{S^i}$ do not increase the largest eigenvalue of $(\sigma^{S^j})_{S^j}$. In the second case, we use technique (\ref{eq:geq}) to show that those $\Pi^{S^i}$ do not increase $(\sigma^{S^j})_{S^j}$. It remains to be shown that $(\sigma^{S^j})_{S^j} \leq \rho_{S^j}$ to apply the Lemma (\ref{lem:monh}) on the monotonicity of the smooth min-entropy. This can be seen inductively from the fact that for every $S^i$ with $i > j$, $\Pi^{S^i}$ decreases the state $\sigma^{S^i}$ either on a larger system if $S^j \subset S^i$ or on the next system in the sequence $(S^k)_{i > k \geq j}$ that contains $S^j$ by fact (\ref{eq:geq}). The proof of estimate (\ref{eq:mpdist}) is then entirely analogous to the two party case.
%\end{proof}
%
%Note that this proof technique only holds for a proper subset $\cK$ of $2^{\{A_1,\dotsc,A_m\}} \setminus \emptyset$ for $m \geq 3$ due to the occurrence of overlapping subsystems.

Note that the smoothing operations $\bar{\cE}^S$ are rank non-increasing and thus $\sigma$ is pure if $\rho$ is. By the Schmidt-decomposition, it follows that if $\rho$ is pure for every pair of subsystems $S, S^c \in \cK$ the application of only one smoothing operation suffices to smooth them both. Therefore, Conjecture \ref{conj:ssc} is also satisfied for tripartite pure states.

\subsection{Three parties}

The last theorem also highlights the simplest case where current techniques fail. Consider a mixed tripartite state $\rho$, $\cK = \{A_1 A_2, A_2 A_3\}$, $\ve > 0$.
%To understand the difficulties of proving conjecture \ref{conj:ssc} for more than two parties we shall consider the example $m = 3$ here with $\cK = \{AB,BC\}$.
Then, proceeding in an iterative way similar to the proof of Theorem \ref{thm:quant2p}, we introduce an ordering $(S^1, S^2) = (A_1 A_2, A_2 A_3)$ and define
$
\sigma \eqdef \bar{\cE}^{A_1 A_2} \circ \bar{\cE}^{A_2 A_3} (\rho)
$
for smoothing quantum operations $\bar{\cE}^{A_1 A_2}$, $\bar{\cE}^{A_2 A_3}$. Ignoring the structure of $\bar{\cE}^{A_2 A_3}$ for the moment, we can always define $\bar{\cE}^{A_1 A_2}$ such that
$
\entHmin(A_1 A_2)_{\sigma} = \entHmin^{\varepsilon,D}(A_1 A_2)_{\bar{\cE}^{A_2 A_3} (\rho)}
$
provided that $\bar{\cE}^{A_2 A_3} (\rho)$ is known. But the question is precisely how to choose $\bar{\cE}^{A_2 A_3}$ so that the application of $\bar{\cE}^{A_1 A_2}$ does not affect the reduced state on $A_2 A_3$ too much?
%Then it remains to define $\Pi^{BC}$ such that the state $\sigma_{BC}$ smooths $\entHmin$ for $\rho_{BC}$. This, however, is when is no known how to continue due to interference between $\Pi^{AB}$ and $\Pi^{BC}$ which generally do not commute. It is unknown whether there exists $\Pi^{BC}$ that bypasses this difficulty and if so how it is defined.
%
\section{Conclusion}
In this note, we presented a simple formulation for a problem that appears as a bottleneck in the analysis of network quantum information processing tasks. We proved that the classical version of the problem can be solved as well as the quantum case when the systems under consideration satisfy a non-overlapping condition. Understanding overlapping marginals seems to be also a barrier in the context of the quantum marginal problem \cite{Kly06}. %The purpose of this note is to give an overview of the state of affairs as far as simultaneous smoothing is concerned to highlight the limits of the current techniques. {\bf (to be completed)}

%% Use \section* for acknowledgement
\section*{Acknowledgment}
The authors would like to thank Fr\'{e}d\'{e}ric Dupuis, Nicolas Dutil, Hamza Fawzi, Patrick Hayden, Renato Renner, Ivan Savov, Pranab Sen, Mark Wilde and Andreas Winter for helpful discussions. This research is supported by the European Research Council grant No. 258932.

%% References:
%% We recommend the usage of BibTeX:
%%
\bibliographystyle{IEEEtranS}
\bibliography{IEEEabrv,qit}
%%
%% where we here have assume the existence of the files
%% definitions.bib and bibliofile.bib.
%% BibTeX documentation can be obtained at:
%% http://www.ctan.org/tex-archive/biblio/bibtex/contrib/doc/
%%
%%
%%
%% Or manual references (pay attention to consistency!):

\end{document}